\documentclass{amsart}

\usepackage{mathrsfs, listings, enumerate, hyperref}

\title{Sachs equations and plane waves III: Microcosms}
\author{Jonathan Holland}
\address{Compunetix\\
  2420 Mosside Blvd \# 1\\
  Monroeville, PA 15146}
\author{George Sparling}
\address{University of Pittsburgh\\
  Department of Mathematics\\
  301 Thackeray Hall\\ Pittsburgh, PA 15260
}
\date{\today}
\newtheorem{theorem}{Theorem}
\newtheorem{lemma}{Lemma}
\newtheorem{corollary}{Corollary}
\newtheorem{definition}{Definition}
\newcommand{\im}{\operatorname{im}}

\newcommand{\op}{\operatorname}

\begin{document}
\begin{abstract}
  This article examines the structure of plane wave spacetimes (of signature $(1,n+1)$, $n\ge 2$) that are homogeneous (the isometry group is transitive) and geodesically complete -- which we call {\em microcosms}.  In general, a plane wave is shown to determine a smooth positive curve in the Lagrangian Grassmannian associated with the $2n$ dimensional symplectic vector space of Jacobi fields.  We show how to solve the Sachs equations in full generality for microcosms and, moreover, we relate the power series expansions of canonical solutions to the Sachs equations on a general plane wave to Bernoulli-like recursions.  It is shown that for microcosms, the curve in the Lagrangian Grassmannian associated to a microcosm is an orbit of a one parameter group in $\op{Sp}(2n,\mathbb R)$.  We also give an effective method for determining the orbit.  Finally, we specialize to the case of $n=2$, and give analytical formulae for the solutions to the Sachs equations and the associated one-parameter group orbit.
\end{abstract}
\maketitle

\tableofcontents

\section{Introduction}
In \cite{SEPI}, the authors introduced and described plane wave spacetimes.  We recollect some relevant definitions.  
\begin{definition}
  A {\em Penrose limit} is a quadruple $(M,G,\mathcal D,\gamma)$ where:
  \begin{itemize}
  \item $M$ is a connected simply-connected manifold, and $G$ is a smooth metric tensor on $M$ of signature $(1,n+1)$;
  \item $\mathcal D_t:M\to M$ is a smooth one-parameter proper group of diffeomorphisms, called the {\em dilation} group;
  \item $\mathcal D_t^*G=e^{2t}G$ for all $t\in\mathbb R$; and
  \item the limit $\lim_{t\to\infty}\mathcal D_t^{-1}x$ exists for all $x\in M$, and accumulates on a curve $\gamma$ which is a smooth curve in $M$.
  \end{itemize}
  Moreover, the curve $\gamma$ is always a null geodesic in $M$, so it is called the {\em central null geodesic}.
\end{definition}

We then define a {\em plane wave} to be a Penrose limit spacetime where the $\mathcal D_t$ and $\gamma$ are ``forgotten''.  Thus a plane wave is a pair $(M,G)$ for which there exists a $\mathcal D_t:M\to M$ making $(M,G,\mathcal D_t,\op{Fix}(\mathcal D))$ a Penrose limit.  In general, every plane wave admits many dilations (at least a $(2n+1)$-dimensional family).  

The unstable manifolds of the dilation group $\mathcal D_t$ of a Penrose limit define the {\em wave fronts}.  In \cite{SEPI}, we show that the wave fronts are $n+1$ dimensional affine spaces on which the dilation acts as a linear transformations fixing an origin defined by the point where the central null geodesic intersects the wave front.  These spaces are ``unstable'' manifolds of $\mathcal D_t$ because they are the directions along which $\mathcal D_t$ expands.  We can thus form an intuitive picture of a Penrose limit as a null geodesic $\gamma$, with affine parameter $u$, and a foliation by wave fronts, one for each $\gamma(u)$ in $\gamma$.

The basic examples of Penrose limits are metrics on the space $\mathbb M = \mathbb U\times\mathbb R\times\mathbb X$ where $\mathbb U$ is a real interval and $\mathbb X$ a fixed real Euclidean space.  We use coordinates $(u,v,x)$ for $\mathbb M$ throughout, where $u\in\mathbb U, v\in\mathbb R$, and $x\in \mathbb X$.

\begin{itemize}
\item A {\em Brinkmann metric}:
  \begin{equation}\label{Brinkmann1}
    \mathcal G_B(p) = 2\,du\,dv + x^Tp(u)x\,du^2 - dx^Tdx
\end{equation}
where $u$ belongs to a real interval, $v\in\mathbb R$, $x$ belongs to a fixed real Euclidean $\mathbb R^n$, and $p(u)$ is an $n\times n$ symmetric real matrix depending smoothly on $u$
  Here the dilation group is
  \begin{equation}\label{BasicDilation}
    \mathcal D_t(u,v,x) = (u,e^{2t}v,e^tx).
  \end{equation}
  The central null geodesic is the curve $v=0,x=0$, $u$ arbitrary.  The wave fronts are the $u=$constant hypersurfaces.
\item A {\em Rosen universe}:
  \begin{equation}\label{RosenUniverse}
  \mathcal G_R(g) = 2\,du\,dv - dx^Tg(u)dx
\end{equation}
where $g(u)$ is a $n\times n$ positive-semidefinite matrix depending smoothly on $u\in\mathbb U$, a real interval.  (The matrix $g$ is positive definite away from a discrete set of points $\mathbb S\subset\mathbb U$, and the points of $\mathbb S$ where $g$ degenerates are removable singularities of the metric.)  The dilation is as in \eqref{BasicDilation}, with the same wave fronts and central null geodesic as in the Brinkmann case.
\item The {\em Alekseevsky metric}:
  \begin{align}
    \label{Alekseevsky1} \mathcal G_\alpha(p,\omega) &= du\,\alpha - dx^Tdx\\
    \notag \alpha &= 2\,dv + x^Tp(u)x\,du - 2x^T\omega(u)\,dx
  \end{align}
where $p(u),\omega(u)$ are smooth $n\times n$ matrices, $p=p^T,\omega=-\omega^T$.  (Again, the dilation, central null geodesic, and wave fronts are as in the other cases.)
\end{itemize}
An important result, ultimately due to Alekseevsky \cite{Alekseevsky}, but reformulated in this way in \cite{SEPI}, is:
\begin{theorem}\label{SEPITheorem1}
Every plane wave is a Brinkmann metric, a Rosen universe, and an Alekseevsky metric.
\end{theorem}

In all of these examples:
\begin{itemize}
\item the central null geodesic is the curve $\gamma : (u\in\mathbb U)\mapsto (u,0_{\mathbb R},0_{\mathbb X})$;
\item the dilation is the group $\mathcal D_t(u,v,x)=(u,e^{2t}v,e^tx)$; and
\item the wave fronts are the $u=$constant hypersurfaces.
\end{itemize}

Every plane wave has a group of at least $2n+1$ symmetries.  For a plane wave with $p(u)$, its Lie algebra of infinitesimal Killing symmetries typically annihilates $u$, and therefore the level sets of $u$ are invariant manifolds on which the Lie algebra integrates to a transitive group of transformations that all fix $u$.  In that case, the Killing algebra generates an isometry group that is transitive on the $2n+1$ dimensional manifold of null directions transverse to any fixed wave front.

In the exceptional case of a plane wave with a group of symmetries that does not fix the $u$ coordinate, the automorphism group is transitive.  They are called {\em homogeneous plane waves}.

The two main examples of homogeneous plane waves are as follows.  First, consider the Alekseevsky metric
\begin{align*}
 \mathcal G_A(\omega, p) &= du\,\alpha - dx^Tdx, \\
 \alpha &= 2\,dv - 2\,x^T\omega\,dx + x^Tpx\,du
\end{align*}
where $p,\omega$ are constant $n\times n$ matrices with $\omega^T=-\omega$, $p^T=p$.  Clearly these have $u\mapsto u+c$ for an extra symmetry, since nothing depends on $u$.  Secondly, consider the Brinkmann metric
\begin{equation}\label{MicrocosmBrink}
  \mathcal G_B(\omega,p) = 2\,du\,dv + x^Te^{-u\omega}pe^{u\omega}x\,du^2 - dx^Tdx.
\end{equation}
(Note that $e^{u\omega}$ denotes the one parameter subgroup in $SO(n)$ generated by the constant skew matrix $\omega$.)  Now the extra symmetry is
$$u\mapsto u+c,\quad x\mapsto e^{-c\omega}x.$$
Every homogeneous plane wave is globally conformal to either of these forms, with conformal factor $e^{bu}$ for a constant $b$.  The obvious change of variables ($X = e^{u\omega}x$) is an isomorphism of Penrose limits.  We summarize:
\begin{lemma}\label{GAGBIso}
  Let $\omega=-\omega^T$ and $p=p^T$ be constant.  Then
  $$\mathcal G_B(\omega,p)\cong\mathcal G_A(\omega,p+\omega^2),$$
  and
  $$\mathcal G_A(\omega,p)\cong\mathcal G_B(\omega,p-\omega^2)$$
  where the isometric isomorphisms preserve the dilation and fix the central null geodesic pointwise.
\end{lemma}

\begin{definition}\label{MicrocosmDef}
A homogeneous plane wave that is geodesically complete is called a {\em microcosm}.
\end{definition}
See \cite{SEPII} for details.  Thus the metrics $\mathcal G_A$ and $\mathcal G_B$ are microcosms, and every homogeneous plane wave is globally conformal to a microcosm.

In the Brinkmann microcosm, the tidal curvature $p(u)$ is the orbit of a one-parameter subgroup of the rotation group.  The purpose of this paper is to obtain an analogous characterization of the Rosen microcosms.  With suitable definitions in place, we shall prove that every Rosen microcosm is the orbit of a one-parameter subgroup of the symplectic group $\op{Sp}(2n,\mathbb R)$.  This theorem is proven in \S\ref{StructureMicrocosm} in general.  

A key ingredient in our construction is algebraic Riccati equation, for an unknown symmetric matrix $X$,
\begin{equation}\label{AlgRiccati1}
  X^2 - \omega X + X\omega + p = 0
\end{equation}
where $\omega=-\omega^T, p=p^T$ are constant real $n\times n$ matrices, but we allow the solution $X$ to be complex.  In \S\ref{SachsEquationsAlekseevsky}, we show how this equation arises naturally from microcosms.  The formalism developed in this section is used in \S\ref{2x2case} to find the solutions in the case relevant to four-dimensional spacetime, when all matrices are $2\times 2$.  Here we can be very explicit about the solutions.  In particular, there are always {\em complex} solutions, and we obtain a characterization of real solutions.  It is significant in our view that the typical case of interest in gravitation, where the energy density of spacetime is positive, there are no real solutions.  Thus complex numbers appear naturally in connection with positive energy.  Although complex numbers are needed to describe the solutions to \eqref{AlgRiccati1}, ultimately the group orbits they describe belong to the {\em real} symplectic group, a phenomenon well-known in algebra.

Algebraic Riccati equations have been studied extensively, often in connection with applications to optimal control.\footnote{A few references are \cite{lancaster1995algebraic}, \cite{laub1979schur}, \cite{paige1981schur}.}
In the optimal control setting, one nearly always assumes that the eigenvalues of the Hamiltonian matrix are not purely imaginary.  Since this excludes some of the most important physical cases of plane waves, we have not been able to find theorems in the literature that map directly onto the results that we need in the proof of Theorem \ref{SpecialRosens}.  While in many applications, it is important that the solutions to the matrix Riccati differential equations should be real, we offer what we hope is a novel perspective that one can still obtain solutions to the differential equation that are real, beginning from complex solutions of the algebraic Riccati equation.  This perspective may in turn have applications to other areas.  We have given complete proofs from scratch.

This paper is organized as follows.

Since it is the conformal structure that governs the abreast Jacobi equation, we focus on the geodesically complete case, i.e., microcosms (Definition \ref{MicrocosmDef}).  In \S\ref{RosenSection}, we recount some features of Rosen metrics that are important to the present article.  A more thorough treatment is given in the earlier article \cite{SEPI}.  A key point of this section is that, while any null geodesic naturally embeds into the Lagrangian Grassmannian of the Jacobi fields abreast of it, for a Rosen metric, the embedding can be written down explicitly.

In \S\ref{SachsEquationsAlekseevsky}, the Sachs and Jacobi equations for microcosms are introduced and studied.  Theorems \ref{TheoremSachsOmega0} and \ref{TheoremSachsOmega} give the general solution of the Sachs equation for microcosms.  These are constructed with the aid of solutions to the Riccati equation \eqref{AlgRiccati1}.  The constant solutions in the case of $n=2$ can be written down by brute force, in \S\ref{2x2case}.

Section \ref{StructureMicrocosm} proves the main result of this paper is Theorem \ref{SpecialRosens}, which is that the embedding of a microcosm into the Lagrangian Grassmannian is an orbit of a one-parameter subgroup of the symplectic group $\op{Sp}(2n,\mathbb R)$.

The final section is devoted to the example case when $n=2$, so for a four-dimensional plane wave spacetime.  We give explicit analytic formulas for the orbit of the symplectic group, and determine the conjugate points.

\section{Preliminaries}
We briefly recall some relevant definitions from \cite{SEPI}.  Let $(M,g)$ be a spacetime, where $M$ is a smooth manifold of dimension $n+2$ and $g$ is a smooth symmetric form of signature $(1,n+1)$ on $M$.  We work at a fixed null geodesic $\gamma$ in $M$.
\begin{definition}
  The space $\mathbb J(\gamma)$ of {\em abreast Jacobi fields} is the space of Jacobi fields along $\gamma$ that are orthogonal to $\gamma$.
\end{definition}
By Jacobi field, we mean a solution of the Jacobi equation
$$\ddot x(u) = P(u)x(u) $$
where $P$ is the tidal curvature operator along $\gamma$, $x$ is a vector field along $\gamma$, and the dot is (covariant) differentiation with respect to the affine parameter $u$ along $\gamma$.  Because the right-hand side of this equation is invariant under translation by a multiple of the tangent vector to $\gamma$, we shall always regard Jacobi fields modulo this tangent vector.    Also, the the tidal curvature along $\gamma$ is the symmetric form on $\gamma^{-1}TM$
$$P(X,Z) = R(\dot\gamma,X,\dot\gamma,Z).$$
(A Jacobi field is called abreast if it is orthogonal to the tangent vector $\dot \gamma$ of $\gamma$.)

We now observe:
\begin{lemma}
  Define a skew-symmetric form on $\mathbb J(\gamma)$ by
  $$\omega(X,Y) = g(\dot X,Y) - g(X,\dot Y).$$
  Then, for each $X,Y\in\mathbb J(\gamma)$, $\omega(X,Y)$ is well-defined, constant along $\gamma$, and non-degenerate (symplectic).
\end{lemma}
The basic arena for this work is then the symplectic vector space $(\mathbb J,\omega)$ (we henceforth suppress the $\gamma$ from the notation).

\subsection{Rosen metrics}\label{RosenSection}
A {\em Rosen metric} is a smooth metric of the form
$$\mathcal G_R = 2\,du\,dv - dx^Th(u)dx.$$
Here $u\in\mathbb U$, a non-empty real interval, $v\in\mathbb R$, $x\in\mathbb X$, and $h(u)$ is a positive-definite symmetric endomorphism of $\mathbb X$, smoothly dependent on $u$.  As in the case of a Brinkmann metric, the vector field $D=2v\partial_v + x\cdot\partial_x$ generates a dilation which fixes the null geodesic $v=0, x=0$.

Every Rosen metric has a change of coordinates that puts it into Brinkmann form.  It is easy to see that the vector fields $\partial_x$ form a Lagrangian basis of Jacobi fields.  Then there exists a factorization of $g(u)$ as $g=L^TL$, with $\dot L = SL$ for some symmetric tensor $S$ satisfying the Sachs equation $\dot S + S^2 + p=0$.\footnote{Note that the decomposition $g=L^TL$ exists globally on $\mathcal G_R$, with $L$ being the components of the Jacobi fields $\partial_x$ in a parallel frame up the central geodesic.}

With $X=Lx$, $V=v+\tfrac12 X^TSX$, we obtain immediately from the equations
$$\mathcal G_R = 2\,du\,dV + X^TpX\,du^2 - dX^TdX.$$
So every Rosen is a Brinkmann, globally in the $u$ coordinate.

Conversely, every Brinkmann is locally a Rosen, by reversing this argument.  On occasion, it is convenient to change freely between Brinkmann and Rosen coordinates.

We think of Rosen metrics as a convenient model for the symplectic space $(\mathbb J,\omega)$.  Specifically, the following is shown in \cite{SEPI}:
\begin{lemma}
The vector fields $\partial_x$ are abreast Jacobi fields for the central null geodesic $\gamma$.
\end{lemma}
Note that there are $n$ independent such vector fields, spanning the $n$-dimensional vector space $\mathbb X$.  Moreover, we have $\omega(\partial_x,\partial_x)=0$, so the space of vector fields $\mathbb X$ is Lagrangian in $\mathbb J$.

Conversely, given an arbitrary spacetime $M$, null geodesic $\gamma$, and Lagrangian subspace $\mathbb X\subset\mathbb J(\gamma)$, we can write down a Rosen metric in any region of the spacetime where the Jacobi fields belonging to $\mathbb X$ are linearly independent.  This Rosen metric can be identified with a Penrose limit of the spacetime around $\gamma$ \cite{penrose1976any}.

\subsection{Embedding in the Lagrangian Grassmannian}
Given a Rosen metric $\mathcal G_R(h)$, for each $u\in\mathbb U$, let $\mathbb H(u)\subset\mathbb J$ be the subspace space of Jacobi fields of $\mathbb J$ vanishing at the point $\gamma(u)$ of the central null geodesic.  Then $\mathbb H(u)$ is a Lagrangian subspace of $\mathbb J$.  This is obvious, by the definition of the symplectic form.  It is also clear on general principles that $\mathbb H$ is smooth.  To see this more directly, a basis of $\mathbb H(u_0)$ consists of the Jacobi fields
$$x^i\partial_v + (H^{ij}(u) - H^{ij}(u_0))\partial_j $$
where $H^{ij}$ is any solution of the differential equation
$$\dot H = h^{-1}.$$

For technical reasons, it shall sometimes be convenient to consider the space $\mathbb J\otimes\mathbb C$ {\em complex} Jacobi fields in this setting, i.e., complex solutions to the differential equation
$$\ddot J + p J = 0.\footnote{Here $p$ is the same {\em real} tidal curvature, and the dot is differentiation with respect to the {\em real} variable $u$ on the null geodesic $\gamma$, and $J$ is a section of $(\gamma^{-1}T\mathbb M/\dot\gamma)\otimes\mathbb C$}$$

A basis $L=[L_1\dots L_n]$ of solutions is Lagrangian provided $\dot L^T L - L^T \dot L=0$.  The Rosen metric associated to a complex Lagrangian subspace is no longer real (and in particular, it is no longer a Lorentzian metric).  However, the set of complex Jacobi fields vanishing at the point $u=u_0$ of the geodesic $\gamma$ is given by
$$z^i\partial_v + (H^{ij}(u) - H^{ij}(u_0))L_j$$
as before, where now $z^i$ are complex coordinates dual to the $L_i$.  This subspace $\mathbb H_{\mathbb C}(u_0) \subset \mathbb J\otimes\mathbb C$ is the complexification of the real subspace of $\mathbb H(u_0)$, i.e.:
\begin{lemma}\label{ComplexificationLemma}
$$\mathbb H_{\mathbb C}(u_0) = \mathbb H(u_0)\otimes\mathbb C \subset\mathbb J\otimes\mathbb C.$$
\end{lemma}
(Indeed, both sides of the equation have the same dimension over $\mathbb C$, whilst the right-hand side is canonically included in the left-hand side.)

\section{Sachs equation}
Because the central null geodesic of a plane wave embeds into the Lagrangian Grassmannian, we can examine it from the point of view of the Sachs equation.  So, as above, let $\mathbb H(t)$ be the Lagrangian subspace of $\mathbb J$ consisting of the Jacobi fields vanishing at the point $u=t$.  Then there is associated a solution of the Sachs equations $S_t=S_t(u)$:
\begin{equation}\label{Steqn}
  \dot S_t + S_t^2 + p(u) = 0.
\end{equation}

\begin{lemma}
$$S_t(u) = (u-t)^{-1}I - \frac13 (u-t)p(t) - \frac14 (u-t)^2p'(t) + O(u-t)^3.$$
\end{lemma}
\begin{proof}
  We showed in \cite{SEPI} that $S_t(u) = (u-t)^{-1}A + O(1)$ where $A$ is the orthogonal projection onto the kernel of $L(t)^T$.  But $L(t)=0$, so $A=I$.  So write
  $$S_t = (u-t)^{-1}I + B + (u-t)C + (u-t)^2D + O(u-t)^3.$$
  We have $\dot S_t = -(u-t)^{-2} + C + 2(u-t)D + O(u-t)^2$, and
  $$S_t^2=(u-t)^{-2} + 2(u-t)^{-1}B + O(1).$$
  Therefore, the Sachs equation gives that $B=0$ because $p$ is regular at $u=t$.  Finally,
  $$\dot S_t + S_t^2 = 3C + 4(u-t)D + O(u-t)^2.$$
  Evaluating the Sachs equation at $u=t$ gives $3C=-p(t)$ and $4D=-p'(t)$.
    
\end{proof}

\subsection{$\omega=0$ microcosm}
Consider now the Sachs equation in the special case that $p$ is a given real constant symmetric matrix:
\[ \dot{S}(u) + S^2(u) + p = 0.\]
We want a solution with $S(u)$ symmetric.  There are at least two clear-cut one parameter families of solutions.  To formulate them rigorously, introduce  the entire holomorphic functions of a complex variable $z$ by the series:
\[ c(z) = \sum_{n = 0}^\infty \frac{(-z)^n}{(2n)!},  \hspace{7pt}   s(z) = \sum_{n = 0}^\infty \frac{(-z)^n}{(2n + 1)!}.\] 
So we have $c(z) = \cos(x)$ and $zs(z) = x\sin(x)$, where $x^2 = z$.  Note that the functions $c(z)$ and $s(z)$ are related by the differential equations:
\[ 2 c'(z) + s(z) = 0, \hspace{7pt} 2zs'(z) + s(z) - c(z) = 0.\]
Here a prime denotes the derivative with respect to $z$.  In terms of the functions $s(z)$ and $c(z)$ define:
\[ T(z) = \frac{s(z)}{c(z)}, \hspace{7pt} U(z) = \frac{c(z)}{s(z)}.\]
Here we take the (complex) domains of $T(z)$ and $U(z)$  to be given by $  |z| < \frac{\pi^2}{4}$.  We have:
\[ z T(z) = x\tan(x), \hspace{7pt} U(z) = x\cot(x), \hspace{7pt} z = x^2, \hspace{7pt} |x| < \frac{\pi}{2}.\]
We have that $T(0) = U(0) = 1$ and the functions $T(z)$ and $U(z)$ obey the differential equations:
\[ 2zT'(z) + T(z) - zT^2(z) - 1 = 0, \hspace{7pt} 2zU'(z) + U^2(z) - U(z) + z = 0.\]
The series for the function $T(z)$ based at the origin begins:
\[ T(z) = 1 + \frac{z}{3} +  \frac{2z^2}{15} + \frac{17z^3}{315} + \frac{62z^4}{2835} + \frac{1382 z^5}{155925} + \frac{21844z^6}{6081075} + \dots.\]
Here we need $|z| < \frac{\pi^2}{4}$ for convergence.
The series for $U(z)$, the reciprocal of $T(z)$, actually convergent for  $|z| < \pi^2$,  begins:
\[ U(z) = 1 - \frac{z}{3} - \frac{z^2}{45} - \frac{2z^3}{945} - \frac{z^4}{4725} - \frac{2z^5}{93555}  -  \frac{1382z^6}{638512875} + \dots. \]
Note that the function $U(z)$ vanishes when $z = \frac{\pi^2}{4}$.  These series have nice expressions in terms of the even Riemann zeta function values, or equivalently in terms of the Bernoulli numbers:
\[ T(z) =  \frac{2}{\pi^2}\sum_{n = 1}^{\infty} (4^n -1 ) \zeta(2n)\left(\frac{z}{\pi^2}\right)^{n-1}  = \sum_{n = 1}^\infty\frac{ 4^n B_{2n}(-z)^{n - 1}(4^n - 1)}{(2n)!}, \hspace{7pt} |z| < \frac{\pi^2}{4}\]
\[ U(z) = 1-2\sum_{n = 1}^{\infty}  \zeta(2n)\left(\frac{z}{\pi^2}\right)^n  = 1 - \sum_{n = 1}^{\infty}  \frac{4^n B_{2n}(-z)^n}{(2n)!},  \hspace{7pt}|z| < \pi^2.\]
Note that $zT(z) = U(z) - U(4z)$, or $U(4z) = U(z) - \frac{z}{U(z)}$.

\begin{itemize} \item  For the  initial condition $S(t) = 0$ we have the solution, denoted $S_t^0(u)$:
\[ S_t^0(u) = p(t - u) T(p(t - u)^2).\]
Less formally, we have $S_t^0(u) = \sqrt{p} \tan\left(\sqrt{p}(t - u)\right)$.\\
For $L_t^0$, such that $\dot{L}_t^0 = L_t^0 S_t^0$, we may take:
\[ L^0_t(u) = \cos\left(\sqrt{p}(t - u)\right) = c\left(p(t - u)^2\right).\]
\end{itemize}
\begin{itemize} 
\item For the solution blowing up as $(u - t)^{-1} I$ as $u \rightarrow t$, denoted $S_t^{\infty}(u)$, we have:
\[  (u - t)S_t^{\infty}(u) = U(p(t - u)^2).\]
Less formally, we have $S_t^{\infty}(u) = \sqrt{p}\cot\left(\sqrt{p}(u - t)\right)$.
For $L_t^{\infty}$, such that $\dot{L}_t^{\infty} = L_t^{\infty} S_t^{\infty}$, we may take, informally:
\[ L^{\infty}_t(u) = (\sqrt{p})^{-1}\sin\left(\sqrt{p}(t - u)\right)\]
Formally this is:
\[ L^{\infty}_t(u)  = (t - u)s(p(t - u)^2).\]
\end{itemize}
Note that $S_t^\infty(u) S_t^0(u) = -p$.  Also note that when $p \ne 0$, such that the trace of $p$ is non-negative, then these solutions will develop singularities as $|u - t|$ increases, controlled by the maximum eigen-value of the matrix $p$. 
\subsection{General Sachs equation}
Now we return to the power description of the completely general Sachs equation $\dot{S}(u) + S^2(u) + p(u) = 0$.  We wish to analyze the series for the solution blowing up as $(u - t)^{-1}$ times the identity as $u \rightarrow t$ for $t$ a basepoint in the domain of the symmetric matrix $p(u)$.  We write out the Taylor expansion for the given symmetric matrix $p(u)$ :
\[ \mathcal{T}(p) =  \sum_{n = 0}^\infty \frac{(u - t)^n}{n!} p_n.\]
So here we have $p_n = \frac{d^n}{du^n} p(u)|_{u = t}$.
Similarly we have the series expansion for $S(u)$:
\[ S(u) = (u - t)^{-1} I - \sum_{n =0}^\infty \frac{(u- t)^n S_n}{n!}.\]
Note that Lemma 5 above gives that $S_0 =  0$, $S_1 = \frac{p_0}{3}$ and $S_2 = \frac{p_1}{4}$.   Then for $S^2$ and $\dot S$ we have:
\[ S^2(u) = (u-t)^{-2} I +  \sum_{n =0}^\infty \frac{(u-t)^n}{n!} \left(- \frac{2S_{n+1}}{n +1} + \sum_{m =0}^n \binom{n}{m}S_m S_{n-m}\right)\]
\[ \dot{S}(u) = - (u-t)^{-2} I - \sum_{n =0}^\infty \frac{(u- t)^{n}}{n!} S_{n + 1}\]
Then the equation $\dot{S}(u) + S^2(u) + p(u) = 0$ gives the  recursion relation, valid for each non-negative integer $n$:
\[ S_{n + 1} = \frac{(n +1)}{(n + 3)}\left(p_n +  \sum_{m =0}^n \binom{n}{m}S_m S_{n - m}\right)\]
Note that by induction each $S_n$ is trivially seen to be symmetric.  When $p$ is constant, so all $p_i$, for $i > 0$ vanish, this recursion relation gives the Bernoulli numbers, as first observed by Euler \cite{Euler}.

\subsection{General solution}
Up until now, we have studied solutions to the Sachs equations corresponding to $L(0)=0$, and those corresponding to $L'(0)=0$, which respectively are the cases in which $S(u)=u^{-1}I+O(u)$ and $S(u) = O(u)$ as $u\to 0$.  We now turn to the problem of determining the general solutions to the Sachs equation on a (Brinkmann) microcosm:
\begin{equation}\label{Sachs1}
\dot S + S^2 + e^{-\omega u}pe^{\omega u} = 0
\end{equation}
where $S$ is an unknown symmetric, complex matrix, $\omega,p$ are constant and real, and $\omega^T=-\omega, p^T=p$.  We first treat the simpler case of $\omega=0$.  In that case, the equation is
\begin{equation}\label{Sachs1a}
\dot S + S^2 + p = 0.
\end{equation}
\begin{theorem}\label{TheoremSachsOmega0}
  Let $\Sigma$ be a constant solution to \eqref{Sachs1a}, so that $\Sigma^2+p=0$, where $p$ is constant.  Then the solution $S$ to \eqref{Sachs1a} in a neighborhood of $u=0$, having initial condition $S_0=S(0)$ a given symmetric (real or complex) matrix, is
\begin{equation}\label{SFromConstant1}
  S = \Sigma + e^{-u\Sigma}(S_0-\Sigma)e^{-u\Sigma}[I + uE(2u\Sigma)e^{-u\Sigma}(S_0-\Sigma)e^{-u\Sigma}]^{-1}.
\end{equation}
where $E$ is the analytic entire function such that $zE(z) = e^z-1$.
\end{theorem}

The significance of this theorem is that the general (real) solution to the Sachs equation can be expressed in terms of a single complex symmetric square root of $-p$.  If $p$ has any positive eigenvalue, then a non-real solution is required.  In particular, in a microcosm of positive energy, so that the trace of $p$ is positive, \eqref{Sachs1} admits a constant non-real complex solution, and every real solution of \eqref{Sachs1} is expressed in terms of complex variables.

\begin{proof}
  We use Lemma \ref{SFromConstant} below, which is for the generic case in which $S_0-\Sigma$ is invertible.    The theorem now follows by continuity (it also follows a direct but tedious calculation).
\end{proof}

\begin{lemma}
  Let $\Sigma$ be a constant solution to \eqref{Sachs1}.  Then every solution $S$ to \eqref{Sachs1} (with constant $p$) in a neighborhood of $u=0$, such that $S(0)-\Sigma$ is invertible, has the form
  \begin{equation}\label{SFromConstant}
    S=\Sigma + [e^{u\Sigma}H_0e^{u\Sigma} + uE(2u\Sigma)]^{-1},\qquad zE(z) = e^z-1
  \end{equation}
  where $H_0=(S(0)-\Sigma)^{-1}$ is determined by the initial condition $S(0)$.
\end{lemma}
\begin{proof}
  Evaluating at $u=0$, we have $S(0)=\Sigma + H_0^{-1}$ and the initial condition is correct, and so it is sufficient to prove that $S$ satisfies the differential equation \eqref{Sachs1}.  To prove this, first rewrite \eqref{SFromConstant} as
  $$S = \Sigma + (LHL^T)^{-1},\qquad L\dot HL^T = I.$$
  Then
  \begin{align*}
    \dot S &= - (LHL^T)^{-1}(\dot L H L^T + L\dot HL^T + LH\dot L^T)(LHL^T)^{-1}\\
           &= - (LHL^T)^{-1}(\Sigma L H L^T + I + LHL^T\Sigma)(LHL^T)^{-1}\\
           &= -(LHL^T)^{-1}\Sigma - \Sigma(LHL^T)^{-1} - (LHL^T)^{-2}\\
           &= - S^2 + \Sigma^2 = -S^2-p
  \end{align*}
  as required.
\end{proof}

\subsubsection{General microcosm}\label{SachsSolGeneral}
We next turn to the general case of an Alekseevsky microcosm $\mathcal G_A(\omega, p)$ with constant $\omega$ (which is skew and not necessarily zero) and $p$ (which is symmetric).  The relevant Sachs equation is
\begin{equation}\label{Sachs2}
\dot S + S^2 + e^{-\omega u}pe^{\omega u} = 0
\end{equation}
where $S$ is symmetric.  If we put $S=e^{-\omega u}Te^{\omega u}$ where $T$ is symmetric, then
$$\dot S = e^{-\omega u}\dot Te^{\omega u} + e^{-\omega u}[T,\omega]e^{\omega u}$$
so that \eqref{Sachs2} is now in the more convenient form
\begin{equation}\label{Sachs2b}
  \dot T + [T,\omega] + T^2 + p = 0.
\end{equation}
We show how to construct a solution to \eqref{Sachs2}.  We later prove that \eqref{Sachs2b} admits constant (complex, symmetric) solutions $T$, so that $[T,\omega]+T^2+p=0$.  Let $T=\Sigma$ be such a constant solution to \eqref{Sachs2b}.  Then, as we have observed, one solution to \eqref{Sachs2} is $S=\sigma=e^{-\omega u}\Sigma e^{\omega u}$.

We shall prove the following:
  \begin{theorem}\label{TheoremSachsOmega}
    Let $\Sigma$ be a constant complex symmetric solution to \eqref{Sachs2b}: $\Sigma^2 + [\Sigma,\omega] + p=0$ where $p,\omega$ are constant $n\times n$ real matrices with $p$ symmetric and $\omega$ skew.  Let $S_0$ a constant symmetric complex matrix such that $S_0-\Sigma$ is invertible.  Then the initial value problem solution $S$ to \eqref{Sachs2}, such that $S(0)=S_0$ is, in a neighborhood of $u=0$:
    \begin{equation}\label{SachsSolutionMicrocosm}
      S(u) = e^{-\omega u}\Sigma e^{\omega u} + e^{-\omega u}e^{(-\Sigma+\omega)u}H(u)^{-1}e^{(-\Sigma-\omega)u}e^{\omega u}
    \end{equation}
    where $H$ is a symmetric complex matrix satisfying $\dot H  = e^{(-\Sigma -\omega)u}e^{(-\Sigma+\omega)u}$ and $H(0) = (S_0-\Sigma)^{-1}$.
  Alternatively, in the case when $S_0-\Sigma$  is not necessarily invertible, let $F=(S_0-\Sigma)H$, so that
  $$F(0) = I, \ \text{and}\ \dot F=(S_0-\Sigma)e^{(-\Sigma -\omega)u}e^{(-\Sigma+\omega)u}$$
  and the solution to \eqref{Sachs2}, such that $S(0)=S_0$, becomes
  \begin{equation}\label{SachsSolutionMicrocosm2}
      S(u) = e^{-\omega u}\Sigma e^{\omega u} + e^{-\omega u}e^{(-\Sigma+\omega)u}F(u)^{-1}(S_0-\Sigma)e^{(-\Sigma-\omega)u}e^{\omega u}.
    \end{equation}
  \end{theorem}
  \begin{proof}
    As in the proof of Theorem \ref{TheoremSachsOmega0}, continuity implies that the second statement of the theorem follows from the first.  So assume $S_0-\Sigma$ is invertible.
    
Let $L$ be a solution to $\dot L = \sigma L = e^{-\omega u}\Sigma e^{\omega u}L$.  Write $L=e^{-\omega u}M$.  Then
$$\dot L = e^{-\omega u}(-\omega M + \dot M) = e^{-\omega u}\Sigma M$$
so that
$$\dot M = (\Sigma + \omega)M.$$
Thus we get a solution
$$L = e^{-\omega u}e^{uA} $$
where $A=\Sigma+\omega$.  Let $H$ be such that $L\dot HL^T=I$.\footnote{In a generic case, we can write
$$H = H_0 + e^{uA}H_1e^{uA^T},\qquad A=\Sigma+\omega$$
with $H_0, H_1$ symmetric. Then
$$\dot H = e^{uA}(AH_1+H_1A^T)e^{uA^T} $$
Assuming we can solve $AH_1+H_1A^T=I$, let $H_1$ be such a solution.  The initial condition $H_0$ is now a free symmetric matrix.  The non-generic case is completely solved in \S\ref{Nongeneric}, although for the argument here, we only need that a solution to $L\dot H L^T=I$ exists with prescribed initial conditions $H(0)$ on a sufficiently small interval, which is obviously the case on any neighborhood of $u=0$ on which $L$ is invertible.}

Next, put
$$\sigma = e^{-\omega u}\Sigma e^{\omega u}, \qquad \dot L = \sigma L, \qquad L\dot HL^T=I$$
$$S = \sigma + (LHL^T)^{-1}.$$
We claim that $\dot S + S^2 + e^{-\omega u}pe^{\omega u} = 0$.  We first observe that $\dot\sigma + \sigma^2 + e^{-\omega u}pe^{\omega u}=0$, which follows from the selection of $\Sigma$.  Next, we have, as above
  \begin{align*}
    \dot S &= \dot\sigma - (LHL^T)^{-1}(\dot L H L^T + L\dot HL^T + LH\dot L^T)(LHL^T)^{-1}\\
           &= -\sigma^2 - e^{-\omega u}pe^{\omega u} - (LHL^T)^{-1}(\sigma L H L^T + I + LHL^T\sigma)(LHL^T)^{-1}\\
           &= - e^{-\omega u}pe^{\omega u} -\sigma^2 -(LHL^T)^{-1}\sigma - \sigma(LHL^T)^{-1} - (LHL^T)^{-2}\\
           &= - e^{-\omega u}pe^{\omega u} - S^2,
  \end{align*}
  as claimed.
\end{proof}

\subsection{Conjugate points}
Let $\gamma$ be a null geodesic in a space-time $M$, affinely parameterized by $u\in\mathbb U$, in the open interval $\mathbb U$ from which we choose a base point $u_0\in\mathbb U$.
\begin{definition}
  Let $\mathbb J(\gamma)$ denote the space of Jacobi fields abreast of $\gamma$.  A point $\gamma(u_1)$ is {\em conjugate} to $\gamma(u_0)$ if there is a non-zero Jacobi field $J\in\mathbb J(\gamma)$ such that $J(u_0)=0$ and $J(u_1)=0$.
\end{definition}
One of the principal applications of the Sachs equations is to give sufficient conditions for the existence of conjugate points, principally via the Raychaudhuri equation, which we show how to formulate in this setting.

Firstly, let $S$ be a solution to the Sachs equation
\begin{equation}\label{SachsConjugatePoints}
  \dot S + S^2 +p = 0
\end{equation}
where as usual $p$ is the tidal curvature along $\gamma$, and the dot is differentiation in the affine parameter $u$, and everything is symmetric.

Define the energy-momentum as
\[E=E(u) = \op{tr}(p(u)).\]
Also, decompose $S$ as $s+\tilde S$, where $s$ is a scalar and $\tilde S$ is tracefree.  Then taking traces on \eqref{SachsConjugatePoints} gives the {\em Raychaudhuri (optical) equation}:
\begin{equation}\label{Raychaudhuri}
  \dot s(u) + s(u)^2 + \op{tr}(\tilde S(u)^2) + E(u) = 0
\end{equation}
which is a first-order scalar differential equation.  Moreover, $\tilde S^2$ is a nonnegative operator, and therefore its trace is nonnegative.

Thus the following:
\begin{lemma}\label{ConjugatePointsLemma1}
  Suppose that $\epsilon$ is such that $\mathbb U_0(\epsilon)=[0, \epsilon^{-1/2}\pi]$ is contained in $\mathbb U$.  If $E(u)>\epsilon>0$ on $\mathbb U_0(\epsilon)$, then $\mathbb U_0(\epsilon)$ contains a point conjugate to $u=0$.
\end{lemma}
\begin{proof}
  We shall show that the maximal solution to \eqref{Raychaudhuri} with $s(u) = Cu^{-1}+O(u)$ as $u\to 0$ must blow up in time $u<\epsilon^{-1/2}\pi$.  By the Raychaudhuri equation,
  \[\dot s \le -s^2 - \epsilon.\]
  Integrating this, with initial condition $\lim_{u\to 0}\cot^{-1}s(u)=0$ gives
  \[s \le \sqrt \epsilon\cot(u\sqrt \epsilon),\]
  and the claimed result follows.

  Now, given a Lagrangian matrix $L$ such that $L(0)=0$ and $\dot L=SL$, we have shown that $\op{tr}(S)$ blows up in time $u<\epsilon^{-1/2}\pi$.  So the matrix $L(u)$ degenerates there, and we are done.
\end{proof}

\section{Structure of microcosms}\label{StructureMicrocosm}

We now analyze the Rosen form of a microcosm, with no restriction on the dimension $n$.  As we mentioned in the last section, the Alekseevsky form is a little bit less convenient because the symplectic form depends on a parameter $\omega$, so we shall use the Brinkmann form.  We proved in \cite{SEPI}, that every microcosm can be put into the form \eqref{MicrocosmBrink}:
\begin{equation}
\mathcal G_B(P_0,\omega) = 2\,du\,dv + x^Te^{-\omega u}P_0e^{\omega u}x\,du^2 - dx^Tdx
\end{equation}
for constant $P_0=P_0^T$ and $\omega=-\omega^T$.

\begin{theorem}\label{SpecialRosens}
  For the microcosm \eqref{MicrocosmBrink}, the curve $\mathbb H(u)$ in the Lagrangian Grassmannian of $\mathbb J$ is an orbit of a one-parameter group of the real symplectic group $\operatorname{Sp}(2n,\mathbb R)$.
\end{theorem}

We recall that a one-parameter group is the group generated by the exponential of a constant Hamiltonian matrix (element of the symplectic Lie algebra).  We shall often refer to {\em one-parameter group orbit} as just a {\em group orbit} when it is clear from context that the orbit is one-dimensional.

The rest of this section contains a proof.  The outline of the proof is to show first that every plane wave is a one-parameter group orbit of the {\em complex} symplectic group.  The construction is slightly easier in the case of ``generic'' conditions, so we do that case first.  Then, invoking Lemma \ref{ComplexificationLemma}, we prove that this is the orbit of a one-parameter group in the real symplectic group.

\subsection{Generic case}
For the solution to the Sachs equation, as in \S\ref{SachsSolGeneral}, we make the ansatz
$$S = e^{-\omega u}S_0e^{\omega u}.$$
Applying this, we get
$$\dot S = -e^{-\omega u}[\omega,S_0] e^{\omega u},$$
so the Sachs equation is
$$-[\omega, S_0] + S_0^2 + P_0=0.$$
This has the type of an {\em algebraic Riccati equation}.  Such equations are frequently encountered in control theory, but the setting is slightly different there.  We shall show in \S\ref{SachsEquationsAlekseevsky} that this equation arises in a perhaps more natural way from an Alekseevsky microcosm.

Consider the Hamiltonian matrix
\begin{equation}\label{Zdefinition}
  Z = \begin{bmatrix}\omega & I \\ -P_0 & \omega\end{bmatrix}
\end{equation}
(where the symplectic form is represented by  $J=\begin{bmatrix}0&-I\\I&0\end{bmatrix}$).
In the control theory setting, it is often assumed that $Z$ has no imaginary eigenvalue.  This assumption implies that the solutions are in some sense controllable.  We shall not make such an assumption, so most of the general results on algebraic Riccati equations in the literature are not quite strong enough for our purposes, and we analyze the relevant structure {\em ab initio}.

There exists a complex symplectic $2n\times 2n$ matrix $U$ such that $U^{-1}ZU=\Sigma$ is upper triangular.  Such a pair $\Sigma, U$ shall be designated a {\em symplectic upper triangularization} of the Hamiltonian matrix $Z$.  In terms of $n\times n$ blocks,
$$\Sigma = \begin{bmatrix}\Sigma_{11}&\Sigma_{12}\\0&\Sigma_{22}\end{bmatrix},\quad U = \begin{bmatrix} U_{11}&U_{12}\\U_{21}&U_{22}\end{bmatrix}=\begin{bmatrix}\alpha&\beta\\\gamma&\delta\end{bmatrix}.$$
\begin{lemma}\label{SymplecticUpperLemma}
  Every $2n\times 2n$ complex Hamiltonian matrix has a symplectic upper triangularization, and there is an effective algorithm for finding it.
\end{lemma}
\begin{proof}
  In the $n=1$ case, we have
  \[Z=\begin{bmatrix}0&1\\-p&0\end{bmatrix}\]
  and
  \[U=\frac{1}{\sqrt{1+q\bar q}}\begin{bmatrix}1&\bar q\\ -q& 1\end{bmatrix}\]
  where
  $$q^2=-p$$
  does the trick.
  
We now show how to go from $n$ to $n-1$.  Represent the symplectic form by a $2n\times 2n$ matrix $J$ with $J+J^T=J^2+I=0$.  Let $Z$ be a complex Hamiltonian matrix (meaning that $Z^TJ+JZ=0$).  The basic idea is to show that $Z$ stabilizes a Lagrangian subspace.  The argument goes as follows.  Let $J(x,y)=x^TJy$.  Let $v\in\mathbb C^{2n}$ be an eigenvector of $Z$ with eigenvalue $\lambda$.  Then $Z$ stabilizes $v$, and the kernel $\ker(v^TJ)$.  Indeed, the first assertion is obvious because $v$ is an eigenvector, and the second is immediate from $J(Zv,x)=\lambda J(v,x) = J(v,Zx)$. so if $J(v,x)=0$, then $J(v,Zx)=0$.  For the inductive step, observe that $J$ descends to a symplectic form on the $2(n-1)$ dimensional space $v^\perp/v$, and $Z\mod v$ defines a Hamiltonian transformation.  Let $W\subset \ker (v^TJ)/v$ be an invariant Lagrangian subspace.  Then $W+v\subset\mathbb C^{2n}$ is an invariant Lagrangian subspace in $2n$ dimensions.

To make this argument effective, it is clearly sufficient to find $2n\times n$ complex matrix $U$ such that the following conditions hold:
\begin{itemize}
\item $ZU=U\Sigma$ for some $n\times n$ matrix $\Sigma$
\item $U^TJU=J$
\item $U^*U = I$
\end{itemize}
where the star denotes the conjugate transpose.  In general, let $A$ be a rectangular complex matrix, and denote by
$$[\ker A], [\im A], [\ker A^T], [\im A^T]$$
the rectangular unitary maps associated with each of the four subspaces, coming from the singular value decomposition of $A$.  Thus, for example: $[\im A]$ is a complex matrix with $r$ columns, where $r$ is the rank of $A$, onto the range of $A$, and such that $[\im A]^*[\im A]=I_r$ and $[\im A][\im A]^*$ is the hermitian projection onto the range of $A$.

Letting $P=[\im v][\im v]^*$ be the projection onto the span of $v$, note that the operator
$$B=(I-P)[\ker v^TJ]$$
has rank $2n-2$, and therefore the matrix $[\im B]$ has $2n-2$ columns: its image is the hermitian complement $\ker(v^TJ)-\im v$.  Now $J_1=[\im B]^TJ[\im B]$ is a symplectic form on $\mathbb C^{2n-2}$, with $Z_1=[\im B]^*Z[\im B]$ a Hamiltonian matrix, and so by induction there is a $2(n-1)\times (n-1)$ matrix $U_1$ satisfying the conditions above decomposition $Z_1U_1=U_1\Sigma_1$, $U_1^TJ_1U_1=0$, $U_1^*U_1=I$.  Finally, $U=[\im v] \oplus [\im B]U_1 $ finishes the construction.
\end{proof}

\begin{lemma}
For a symplectic matrix $\begin{bmatrix}\alpha&\beta\\\gamma&\delta\end{bmatrix}$ with $\alpha$ invertible, $\gamma\alpha^{-1}$ is symmetric.
\end{lemma}
\begin{proof}
For a symplectic matrix, $\gamma^T\alpha=\alpha^T\gamma$.  Therefore, $\alpha^T\gamma\alpha^{-1} = \gamma^T = \alpha^T(\gamma\alpha^{-1})^T$.  So $\gamma\alpha^{-1}$ is symmetric.
\end{proof}

\begin{lemma}\label{SpecialInvariantSubspace}
  Consider a complex matrix $Z$ of the form
  $$Z = \begin{bmatrix}A&B\\ C&-A^T\end{bmatrix} $$
  where $B=B^T, C=C^T$ are symmetric, and $B$ is invertible.  Then $Z$ has an invariant complex Lagrangian subspace complementary to the Lagrangian subspace $\mathbb L = \im\begin{bmatrix}0\\ I\end{bmatrix}$.
\end{lemma}
\begin{proof}
  We proceed by induction on $n$, the base case $n=0$ being trivial.  To go from $n$ to $n-1$, we observe that at least one of the first $n$ components of an eigenvector $v$ of $Z$ must be non-zero, because if $v=0\oplus v_2$, 
  with nonzero $v_2\in\mathbb C^n$, then
  $$Zv = \begin{bmatrix}Bv_2\\ *\end{bmatrix} = \begin{bmatrix}0\\v_2\end{bmatrix}\lambda$$
  so $Bv_2=0$, which contradicts the invertibility of $B$.  So we can suppose that $v=\begin{bmatrix}v_1\\v_2\end{bmatrix}$ where $v_1\ne 0$, and $v^*v=1$.

  Let
  $$P=\begin{bmatrix}I_n&0\\ 0&0\end{bmatrix},\qquad Q=I-P=\begin{bmatrix}0&0\\ 0&I_n\end{bmatrix}.$$
  The condition of the lemma is that $PZQ$ has rank $n$.
 
  Let $V = I-vv^*$ be the hermitian projection onto the complement of $v$, and let $J_1=V^TJV$ be the pullback of the symplectic form onto the image of $V$, and $Z_1=V^*ZV$ the pullback of $Z$ onto the image of $V$.  We have to show that $Z_1$ satisfies the conditions of the lemma, namely that $PZ_1Q$ has rank $n-1$.

  We find the non-zero (top right) block of $PZ_1Q$ is
  $$[PZ_1Q]_{12} = [(-1+v_1v_1^*)A+v_1v_2^*C]v_1v_2^* + [(1-v_1v_1^*)B+v_1v_2^*A^T](1-v_2v_2^*). $$
  Because $v$ is an eigenvector of $Z$, we have
  $$Av_1 = \lambda v_1 - Bv_2,$$
  and so
  \begin{align*}
    [PZ_1Q]_{12} &= (-1+v_1v_1^*)(\lambda v_1-Bv_2)v_2^*+v_1v_2^*Cv_1v_2^* + [(1-v_1v_1^*)B+v_1v_2^*A^T](1-v_2v_2^*)\\
                 &= (1-v_1v_1^*)B + v_1v_2^*[Cv_1v_2^* + A^T(1-v_2v_2^*)].
  \end{align*}
  Applying $(1-v_1v_1^*)$ on the left gives
  $$(1-v_1v_1^*)[PZ_1Q]_{12} = (1-v_1v_1^*)B$$
  which is the product of $n\times n$ matrices of respective ranks $n-1$ and $n$.  Therefore it has rank $n-1$ and we are done by induction.
  
\end{proof}

\begin{lemma}
  Consider the Sachs equation, for the unknown symmetric matrix $X$
  \begin{equation}\label{AlgSachs}
    -[\omega, X] + X^2 + P = 0 
  \end{equation}
  where $\omega, P$ are given matrices with $\omega$ skew and $P$ symmetric.  Then \eqref{AlgSachs} admits a (possibly complex) solution $X$.
\end{lemma}
\begin{proof}
  By Lemma \ref{SpecialInvariantSubspace}, the Hamiltonian matrix $Z$ stabilizes a Lagrangian subspace complementary to $\op{im}\begin{bmatrix}0\\I\end{bmatrix}$.  Therefore symplectic upper triangularization $U$ constructed in Lemma \ref{SymplecticUpperLemma}, using this invariant subspace, is such that $\alpha=U_{11}$ is invertible.  That is, we have
  \begin{align*}
    \gamma &= \alpha\Sigma_{11} - \omega\alpha\\
    \omega\gamma &= P_0\alpha + \gamma\Sigma_{11}
  \end{align*}
  where $\alpha$ is an invertible $n\times n$ complex matrix.
  
  Let $X=\gamma\alpha^{-1}$.  We claim that $X$ solves \eqref{AlgSachs}.  We have
  $$\omega X = \omega\gamma\alpha^{-1} = P_0 + \gamma\Sigma_{11}\alpha^{-1}.$$
  Also,
  $$X^2 = \gamma\alpha^{-1}(\alpha\Sigma_{11}-\omega\alpha)\alpha^{-1} = -X\omega + \gamma\Sigma_{11}\alpha^{-1}.$$
  Thus
  $$X^2-\omega X = - X\omega - P_0,$$
  which becomes \eqref{AlgSachs} after moving terms to the left-hand side.

\end{proof}

For the rest of this subsection, assume that we are in the following generic case:
\begin{definition}
  A solution $X=S_0$ of \eqref{AlgSachs} is called {\em generic} if $\omega - S_0$ and $\omega + S_0$ share no eigenvalue.
\end{definition}
Note that genericity is a condition on $S_0,$ {\em and} $\omega,P_0$, but not separately in $S_0$ alone.  For example, when $\omega=P_0=0$, then there are no ``generic'' solutions $S_0$.

To formulate the genericity precisely, let $\mathbb M$ be the space set of $n\times n$ complex matrices, and $\mathbb X^2\subset\mathbb M$ the space of complex symmetric matrices.  Consider the mapping $F:\mathbb M\to\mathbb X^2$ given by
$$F(M) = k[M,M^T] + (M+M^T)^2$$
where $k$ is an arbitrary constant.  Let $\mathbb U\subset\mathbb M$ denote the Zariski open set of matrices $M$ having no eigenvalues $\lambda,\mu$ such that $\lambda+\mu=0$.  Then:
\begin{lemma}
The image of $\mathbb M$ under $F$ contains a nonempty Zariski open subset of $\mathbb X^2$.
\end{lemma}
\begin{proof}
  The differential of $F$ at the identity $I\in\mathbb U$ is:
  $$dF(I) = 4(dM + dM^T).$$
  Since this has full rank, there is a Zariski open neighborhood of the identity on which $F$ is an open mapping.
\end{proof}

In the generic case, Theorem \ref{SpecialRosens} can be proven by an explicit construction.  Recalling that
$$S=e^{-\omega u}S_0e^{\omega u},\qquad -[\omega,S_0] + S_0^2 + P_0 = 0.$$
We now solve the equation $\dot L=SL$ with 
$$L = e^{-\omega u}e^{(S_0+\omega)u}L_0.$$

To get the Rosen form, we seek a symmetric tensor $H$ such that
$$\frac{dH}{du} = (L^TL)^{-1}.$$
We try, as in \ref{SachsSolGeneral}, with $A=-\omega-S_0$,
$$H = L_0^{-1}e^{Au}L_0H_0L_0^Te^{A^Tu}L_0^{-T}.$$
Then
$$\frac{dH}{du} = L_0^{-1}e^{Au}\left(AL_0H_0L_0^T + L_0H_0L_0^TA^T\right)e^{A^Tu}L_0^{-T},$$
which is $(L^TL)^{-1}$ provided that the term in brackets is the identity.  Generically, $A,-A^T$ do not share an eigenvalue, and so the function $F(X) = AX + XA^T$ is an isomorphism from the space of symmetric matrices to itself, and therefore there exists a solution $L_0H_0L_0^T=X$ of $F(X) = I$.\footnote{
  Lemma: $F(X) = AX + XA^T$ is invertible as linear operator $F:\mathbb X^2\to\mathbb X^2$ provided $A$ has no eigenvalues $\lambda,\mu$ such that $\lambda+\mu=0$.
  Proof:
    Suppose that $AX=-XA^T$.  Then, for any polynomial $p$, we have $p(A)X = Xp(-A^T)$.  If $p$ is the characteristic polynomial of $A$, then we have $Xp(-A^T)=0$, and $p(-A^T)$ is invertible by hypothesis, so $X=0$.
}  

  With this choice of $H_0$, we have
  $$\begin{bmatrix}I\\H\end{bmatrix} \sim \begin{bmatrix}(L_0^{-1}e^{uA}L_0)^{-T}\\ (L_0^{-1}e^{uA}L_0)H_0\end{bmatrix} $$
  i.e., $H$ belongs to the orbit of a one-parameter subgroup of $\operatorname{Sp}(\mathbb J\otimes\mathbb C)$ (which happens to be block diagonal in this case).

{\em Generically}, the equation $AL_0H_0L_0^T + L_0H_0L_0^TA^T=I$ is solvable for $L_0H_0L_0^T$, but of course not always.  To analyze the obstruction in the remaining cases, we may as well take $L_0=I$.
We look for a modification
$$H = e^{uA}H_0e^{uA^T} + M(u),\qquad \dot H = e^{uA}e^{uA^T}$$
where $H_0$ is a constant to be determined and $M(u)$ is an unknown function of $u$.  All told, we have
$$e^{uA}(I - AH_0 - H_0A^T)e^{uA^T} - \dot M = 0.$$
We show how to solve this equation.  Equip the space $\mathbb X^2$ of symmetric $n\times n$ matrices with the trace form.  Consider the linear operator on $\mathbb X^2$ given by $F(X) = AX+XA^T.$  Note that $F^T(Y)=A^TY + YA$ is the adjoint with respect to the trace form.  If we didn't have the $\dot M$ term, the equation for $H_0$ that we would need to solve is
$$I - F(H_0) = 0.$$
By standard linear algebra, there exist $H_0\in\mathbb X^2, M_0\in\ker(F^T)$ such that
$$I - F(H_0) - M_0=0$$
with $M_0$ unique.  Explicitly, let $P$ be the orthogonal projection on $\mathbb X^2$ onto the image of $F$, and $Q=\mathcal I_{\mathbb X^2}-P$ the orthogonal projection onto the kernel of $F^T$, then:
$$M_0=Q(I).$$
If $M$ is such that
$$\dot M = e^{Au}M_0e^{A^Tu},$$
with $F^T(M_0) = A^TM_0+M_0A=0$, then we have our solution.  (This is the generic case of the theorem, and we are done.)

\subsection{Non-generic case}\label{Nongeneric}
We now drop the assumption of genericity, continuing the discussion of the preceding subsection.

Suppose we have selected $H_0,M_0$ such that
\begin{equation}\label{AH0M0}
  I = AH_0 + H_0A^T + M_0,\quad A^TM_0 + M_0A=0.
\end{equation}
For $w$ an unknown constant symmetric matrix, put
$$W=\begin{bmatrix}-A^T&0\\ w&A\end{bmatrix}$$
and
$$e^{uW} = \begin{bmatrix}e^{-uA^T}&0\\ \gamma &e^{Au}\end{bmatrix}.$$
Hence
$$\dot\gamma = -\gamma A^T + e^{uA}w.$$
Also, let $H=(\gamma + e^{uA}H_0)e^{uA^T}$ so that
$$\begin{bmatrix}I\\H\end{bmatrix}\sim e^{uW}\begin{bmatrix}I\\ H_0\end{bmatrix}.$$
We want $\dot H=e^{uA}e^{uA^T}$.  Now
$$\dot H = \dot\gamma e^{uA^T} + \gamma A^T e^{uA^T} + e^{uA}(AH_0+H_0A^T)e^{uA^T}$$
so we need
\begin{align*}
  \dot\gamma e^{uA^T} + \gamma A^T e^{uA^T} &= e^{uA}M_0e^{uA^T}\\
  \dot\gamma + \gamma A^T &= e^{uA}M_0\\
  e^{uA}w &= e^{uA}M_0\\
  w &= M_0.
\end{align*}
Therefore we have proven that $H$ belongs to an orbit of the one-parameter group generated by
$$W =
\begin{bmatrix}
  -A^T & 0 \\
  M_0 & A
\end{bmatrix}.
$$

\subsection{Recovering reality}
Finally, to complete the proof, we must show that the one parameter group can be taken to be real.  The problem is that the Sachs tensor $S$ is complex, corresponding to a complex polarization of the symplectic space $\mathbb J$.  The propagator is thus apparently complex, being written in a complex basis.  In fact, it is clear that the Rosen metric, with $g(u)=L^TL$ must also be complex when it is computed relative to the Lagrangian basis determined by $S$.  However, this is precisely the case contemplated in Lemma \ref{ComplexificationLemma}, which we now re-frame as:

\begin{lemma}
Suppose that the subspace $\mathbb H_{\mathbb C}(u)$ is an orbit of $Z\in\mathfrak{sp}(2n,\mathbb C)$.  Then $e^{uZ}\mathbb H_{\mathbb C}\subset\mathbb J\otimes\mathbb C$ is the complexification of a real space $\mathbb H(u)\subset\mathbb J$.
\end{lemma}

We must show that the curve in the Grassmannian $\mathbb H(u)\subset\mathbb J$, which is known to be both real and the orbit of a $Z\in\mathfrak{sp}(2n,\mathbb C)$, is in fact the orbit of the real part of $Z$.  This is true in more generality: let $G$ be a real Lie group with Lie algebra $\mathfrak g$, $G_{\mathbb C}$ its complexification, with $\mathfrak g_{\mathbb C}=\mathfrak g\otimes\mathbb C$ and $G\subset G_{\mathbb C}$ the canonical inclusion.  Let $L\subset G_{\mathbb C}$ be a closed subgroup, with Lie algebra $\mathfrak l$.  We must show that, if $Z\in\mathfrak g_{\mathbb C}$ is such that
$$\exp(tZ)L = \exp(t\bar Z)L $$
then, with $X=\frac{Z+\bar Z}{2}$,
$$\exp(tX)L = \exp(tZ)L.$$

We use the following criterion from Goto \cite{goto1970orbits}.  For $Z\in\mathfrak g_{\mathbb C}$, define
$$ P_L(Z) = \{ A\in\mathfrak l\,|\, (\operatorname{ad}Z)^nA \in \mathfrak l, n=0,1,2,\dots\}.$$
Then $\exp(tZ)L=\exp(tW)L$ for all $t\in\mathbb R$ if and only if $Z-W\in P_L(Z)$.  From the assumption $\exp(tZ)L=\exp(t\bar Z)L$, one direction of Goto's criterion gives $Z-\bar Z\in P_L(Z)$.  With
$$X = \frac{Z+\bar Z}{2},$$
we have
$$X - Z = \frac{\bar Z - Z}{2}\in P_L(Z)$$
and so the other direction of Goto's criterion now gives
$$\exp(tX)L = \exp(tZ)L$$
for all $t\in\mathbb R$.  Thus the orbit is that of a real group.

Specializing this to the the present case of $G=\operatorname{Sp}(2n,\mathbb R)$ and $L_{\mathbb R}\subset G$, $L\subset G_{\mathbb C}$ the respective stabilizers of {\em real} Lagrangian subspace $\mathbb H$, we obtain that the trajectory $\exp(tZ)L=\exp(tX)L$ in $G_{\mathbb C}/L$ is exactly the trajectory $\exp(tX)L_{\mathbb R}$ in $G/L_{\mathbb R}$ under the canonical inclusion of the latter in the former.

This completes the proof of Theorem \ref{SpecialRosens}.

\subsection{Examples}
We give a collection of examples illustrative of Theorem \ref{SpecialRosens} and its proof.  

1. Consider the case where $\omega=0$, $P=P_0=I$.  The Jacobi equation is
$$\ddot L + L = 0,$$
with $L(0)$ invertible, $\dot LL^{-1}$ symmetric.  Two obvious solutions, one obviously real and the other obviously non-real, are
$$L_{\mathbb R} = \cos u I,\qquad L_{\mathbb C} = e^{iu}I.$$
We will show how the complex solution is secretly real.  We have
$$H_{\mathbb R} = \tan u I,\qquad H_{\mathbb C} = (-2i)^{-1}e^{-2iu}I.$$
The content of the theorem is that there exists a constant complex symplectic matrix $S$ and $A=A(u)$ in $GL(n,\mathbb C)$ such that
\begin{equation}\label{SymplecticChangeComplex1}
  S\begin{bmatrix}I\\ H_{\mathbb C}\end{bmatrix} = \begin{bmatrix}I\\ H_{\mathbb R}\end{bmatrix}A(u).
\end{equation}
Up to a constant complex factor, put $S=\begin{bmatrix}-\tfrac12I&iI\\\tfrac12 iI & -I\end{bmatrix}$.  Then \eqref{SymplecticChangeComplex1} is seen to follow, for a suitable factor $A(u)$ from the observation:
$$S\begin{bmatrix}(-2i)e^{2iu}\\1\end{bmatrix} = \begin{bmatrix}i(e^{2iu}+1)I\\ (e^{2iu}-1)I\end{bmatrix}.$$

2. Consider a case where $A=\omega-S_0$ has a degenerate repeated zero eigenvalue
$$A=\begin{bmatrix}0&1\\0&0\end{bmatrix}.$$
The differential equation we have to solve is
$$\dot H  = e^{uA}e^{uA^T} = \begin{bmatrix}1&u\\0&1\end{bmatrix}\begin{bmatrix}1&0\\u&1\end{bmatrix}=\begin{bmatrix}u^2+1&u\\u&1\end{bmatrix}$$
so that
$$H = \begin{bmatrix}\frac{u^3}{3} + u& \frac{u^2}{2}\\\frac{u^2}{2} & u\end{bmatrix} + H_0.$$
We show how to select the initial condition $H_0$ and verify that $H$ is an orbit of the symplectic group using the procedure in the proof of the theorem.

First, solve the equation
$$A^TM_0 + M_0A = 0,$$
with
$$M_0 = \begin{bmatrix}0&0\\0&1\end{bmatrix}.$$
Note that $AM_0=A$, $M_0A^T=A^T$.  Also, put
$$H_0=\tfrac12\begin{bmatrix}0&1\\1&0\end{bmatrix}$$.  Then we have $AH_0+H_0A^T + M_0=I$, so that we have constructed a solution to \eqref{AH0M0}.
Finally, consider
$$W = \begin{bmatrix}-A^T & 0\\ M_0 & A\end{bmatrix}.$$
Then
$$\exp(uW) = \begin{bmatrix}1&0&0&0\\ -u&1&0&0\\ -u^3/6+u&u^2/2 &1&u\\ -u^2/2& u&0&1\end{bmatrix}=\begin{bmatrix}S_{11}&S_{12}\\S_{21}&S_{22}\end{bmatrix}.$$
We then see that
$$\exp(uW)\begin{bmatrix}I\\H_0\end{bmatrix}H = \begin{bmatrix} 1 & 0 \\
 -u & 1 \\
 \frac{u}{2}-\frac{u^3}{6} & \frac{u^2}{2}+\frac{1}{2} \\
 \frac{1}{2}-\frac{u^2}{2} & u
\end{bmatrix} = \begin{bmatrix} I \\ H \end{bmatrix}\begin{bmatrix} 1 & 0 \\
 -u & 1
\end{bmatrix}$$
as required.

3. Consider a case where $A=\omega-S_0$ has a pair of opposite eigenvalues $\lambda=\pm1$, say
$$A=\begin{bmatrix}1&0\\0&-1\end{bmatrix}.$$
Here we can take $M_0=0$ and
$$H_0=A/2.$$
Then we have to solve
$$\dot H = e^{uA}e^{uA^T} = e^{2uA}.$$
Thus, again omitting constants,
$$H = 2^{-1}\begin{bmatrix}e^{2u}&0\\0&-e^{-2u}\end{bmatrix}.$$
On the other hand, with
$$W = \begin{bmatrix}-A^T&0\\0&A\end{bmatrix},$$
we have
$$\exp(uW)\begin{bmatrix}I\\H_0\end{bmatrix} = \begin{bmatrix}e^{-u}&0\\0&e^u\\ \frac{1}{2}e^{u}&0\\0&\frac{1}{2}e^{-u}\end{bmatrix} = \begin{bmatrix}I\\H\end{bmatrix}\begin{bmatrix}e^{-u}&0\\0&e^u\end{bmatrix}. $$

4. For another example, suppose
$$A=\begin{bmatrix}2&-3\\1&-2\end{bmatrix}.$$
The projection $Q(I)$ of the identity matrix onto the kernel of $F^T(X) = A^TX+XA$ is
$$M_0=Q(I) = \frac29 \begin{bmatrix}1&-2\\ -2&3\end{bmatrix}.$$
Also, we see that
$$H_0=-\frac{1}{54}\begin{bmatrix}0&7\\7&8\end{bmatrix}$$
is such that
$$AH_0 + H_0A^T + M_0 = I.$$
Putting
$$W = \begin{bmatrix} -A^T & 0 \\ M_0 & A\end{bmatrix},$$
we have
$$\exp(uW)\begin{bmatrix}I\\ H_0\end{bmatrix} = \begin{bmatrix}W_1\\ W_2\end{bmatrix} = \begin{bmatrix}I\\H\end{bmatrix}W_1 $$
where we find with the aid of Mathematica that
$$H(u) = W_2W_1^{-1} = \frac{1}{108}\begin{bmatrix}
 27 \left(-24 u-5 e^{-2 u}+9 e^{2 u}-4\right) & -432 u-135 e^{-2 u}+81 e^{2 u}+40 \\
 -432 u-135 e^{-2 u}+81 e^{2 u}+40 & -216 u-135 e^{-2 u}+27 e^{2 u}+92 \\
\end{bmatrix}
.$$

And (again with Mathematica), we find
$$\dot H = e^{uA}e^{uA^T}.$$

\section{Alekseevsky metrics}\label{2x2case}
Until now, we have looked at the Sachs equation for a Brinkmann metric.  Now we consider instead the Alekseevsky metric $\mathcal G_A(\omega, p)$ \eqref{AlekseevskyConst} whose inputs are constant real $2\times 2$ matrices $\omega=-\omega^T, p=p^T$.  With these data, we shall show that the Sachs equation for the (symmetric) shear tensor $S$ is:
$$\dot S + S^2 - [\omega,S] + p -\omega^2 = 0.$$

As an application, we show how to find the constant (complex) solutions to the Sachs equation for a four-dimensional microcosm ($n=2$).  Then we also determine the symplectic group orbit from this solution.  Our emphasis is on providing explicit analytical formulae.  Finally, we shall derive a formula for the conjugate points along a geodesic, and show that there are infinitely many conjugate points if and only if $p-\omega^2$ has a positive eigenvalue.

\subsection{Sachs equation}\label{SachsEquationsAlekseevsky}

We specialize again to the case of an Alekseevsky metric
\begin{equation}\label{AlekseevskyConst}
\begin{array}{rl}
 \mathcal G_A(\omega,p) &= du\,\alpha - dx^Tdx, \\
 \alpha &= 2\,dv - 2\,x^T\omega\,dx + x^Tpx\,du
\end{array}
\end{equation}
(Here $\omega$ is skew and $p$ is symmetric $n\times n$ real matrices, both real and constant, and $x\in\mathbb R^n$.)

We collect the main results of this section into two theorems:
\begin{theorem}
The Jacobi equation for the metric $\mathcal G_A(\omega,p)$ is
$$\ddot X - 2\omega\dot X + pX = 0,$$
with symplectic form given by
$$\Omega(X,Y) = X^T \dot Y - Y^T \dot X - 2X^T\omega Y.$$
The Sachs equations are
    \begin{equation}\label{SachsAlekseevsky}
      \begin{array}{l}
        \dot S + S^2 - [\omega,S] + p - \omega^2 = 0\\
        \dot L = (S+\omega)L.
      \end{array}
    \end{equation}
\end{theorem}
(This theorem is the combination of Theorem \ref{AlekseevskyToRosen} and Lemma \ref{AlekseevskySymplectic} below.)

\subsection{Conversion to Rosen form}
This section recounts one of the constructions of \cite{SEPI} which is particularly perspicacious in the setting of an Alekseevsky microcosm $\mathcal G_A(\omega,p)$ given as in \eqref{AlekseevskyConst}.  We show how to construct the Rosen form of the metric in this case.  To this end, consider coordinate changes of the form:
\begin{equation}\label{rosenCoordSys}
  v = V + 2^{-1}x^TS(u)x,  \quad  x = L(u)X
\end{equation}
where $S, L$ are given endomorphisms of $\mathbb X$ (smooth in $u$), with $S$ symmetric and $L$ invertible.
\begin{theorem}\label{AlekseevskyToRosen}
  The following are equivalent:
  \begin{itemize}
  \item The metric \eqref{AlekseevskyConst} transforms to a Rosen metric under \eqref{rosenCoordSys}:
    \begin{equation*}
      \mathcal G_A = 2\,du\,dV - dX^Th(u)dX, \quad \text{with\ } h(u)=L(u)^TL(u);
    \end{equation*}
  \item $L$ and $S$ satisfy the system of Sachs equations \eqref{SachsAlekseevsky}:
    \begin{equation}
      \begin{array}{l}
        \dot S + S^2 - [\omega,S] + p - \omega^2 = 0\\
        \dot L = (S+\omega)L.
      \end{array}
    \end{equation}
  \end{itemize}
\end{theorem}

\begin{proof}
Let $B$ be such that
\[ \dot L= BL.\]  We want the metric $\mathcal G_A(\omega,p)$ to be of Rosen form in the coordinates $V,X$.  
In particular, the $du\,dX$ term is zero, which is
$$-dX^TL^TBx + dX^TL^TSx - 2 x^T\omega L\,dX - x^TB^TL\,dX + x^TSL\,dX=0$$
$$- 2 x^T\omega L\,dX - 2x^TB^TL + 2x^TSL=0$$
$$-\omega - B^T + S=0.$$
From the last equation, since $S$ is symmetric and $\omega$ is skew, we find
$$S=B_\odot,\quad \omega=B_\wedge.$$
The $du^2$ term must also vanish, giving (with $B=S+\omega$):
\begin{align*}
  0 &= p - \omega B + B^T\omega + SB + B^TS - B^TB + \dot S\\
    &= p - \omega^2 - [\omega,S] + S^2 + \dot S.
\end{align*}

So \eqref{SachsAlekseevsky} now follows directly.
\end{proof}

\subsection{Jacobi equation and symplectic form}
Differentiating the second equation of \eqref{SachsAlekseevsky},
\begin{align*}
  \ddot L &= -(S^2 - [\omega,S] + p - \omega^2)L + (S+\omega)^2L\\
          &=(-p + 2\omega^2+2\omega S)L\\
          &=2\omega\dot L - pL.
\end{align*}
So the Jacobi equation is
\begin{equation}\label{AlekseevskyJacobi}
  \ddot X - 2\omega\dot X + pX = 0.
\end{equation}

\begin{lemma}\label{AlekseevskySymplectic}
The symplectic form on solutions $X,Y$ to \eqref{AlekseevskyJacobi} is
$$\Omega(X,Y) = X^T \dot Y - Y^T \dot X - 2X^T\omega Y.$$
\end{lemma}
\begin{proof}
  We have to show that the derivative of $\Omega$ with respect to $u$ is constant.  We have
  $$\ddot X = 2\omega\dot X - pX,\qquad \ddot Y = 2\omega\dot Y - pY.$$
  We have, taking into account the symmetry of $p$ and antisymmetry of $\omega$,
  \begin{align*}
    \frac{d}{du}\Omega(X,Y)
    &= X^T \ddot Y - Y^T \ddot X - 2\dot X^T\omega Y - 2X^T\omega \dot Y\\
    &= 2X^T \omega\dot Y - 2Y^T\omega\dot X + 2 Y^T\omega\dot X - 2X^T\omega\dot Y = 0.
  \end{align*}
\end{proof}

We put:
\[ Z = \begin{bmatrix} L\\\dot L\end{bmatrix}\quad .\]
Then we get:
\[ \dot Z =   \begin{bmatrix} \dot L\\\ddot L\end{bmatrix} =    \begin{bmatrix} \dot L\\2\omega \dot L  - pL\end{bmatrix} =    \begin{bmatrix} 0&I\\-p&2\omega\end{bmatrix}    \begin{bmatrix} L\\\dot L\end{bmatrix}, \]
\[ \dot Z = QZ, \quad  Q =   \begin{bmatrix} 0&I\\-p&2\omega\end{bmatrix}.\]

\subsection{Solutions for $n=2$}
We now specialize to the case of $n=2$ for the rest of the section.  Put
$$\omega = \begin{bmatrix}0&-w\\w&0\end{bmatrix},\qquad p = \begin{bmatrix} A+B& C\\ C& A-B\end{bmatrix},\qquad S = \begin{bmatrix}s+u&t\\ t&s-u\end{bmatrix}.$$
Observe that the Einstein stress-energy-momentum is represented by the scalar $E=\op{tr}(p-\omega^2) = 2(A+w^2)$.  Thus the energy density is positive if $A > -w^2$, negative if $A < -w^2$, and vacuum if $A=w^2$.

We now examine the Sachs equation, for {\em constant} $S$:
\[S^2 - [\omega,S] + p - \omega^2 = 0.\]
Taking traces, we see that if there is a \textbf{real solution} for $S$, then $\textrm{tr}(S^2) \ge 0$, so $E  \le 0$, so the energy must be negative or zero.  The Sachs equation expands to
\[\begin{bmatrix}X-Y & Z\\ Z& X+Y\end{bmatrix} = 0\]
\[X = A + s^2 + t^2+u^2+w^2,\quad Y = B+2su+2tw,\quad Z = C+2st-2uw.\]
So the triple $(s, t, u)$ is required to solve the following equations, given the real numbers $A, B,  C$ and $w$:
\[  B + 2su + 2zt=0, \quad C + 2st - 2wu =0, \quad  A + s^2 + z^2 + t^2 + u^2=0.\]

When these equations hold, we note, in particular,  the relations:
\[ B^2 + C^2 =  4(st + wu)^2 + 4(su - wt)^2 = 4(s^2 + w^2)(t^2 + u^2), \]
\[  A^2 - B^2 - C^2 = \left(s^2 + w^2   -  t^2 -u^2\right)^2.\]
There is always a complex solution.  If $B^2+C^2 > 0$, then $s^2+w^2\ne 0$, and the unique solution of
$$B + 2su + 2wt=0, \quad C +  2st - 2wu=0$$
is
$$u = \frac{-Bs+Cw}{2(s^2+w^2)},\quad t = \frac{-(Cs + Bw)}{2(s^2+w^2)}.$$
Then $A^2 = B^2 + C^2 + \left(s^2 + w^2   -  t^2 -u^2\right)^2$ is a quadratic in $(s^2+w^2)^2$ with nonzero constant term, and therefore admits a nonzero solution.  

\subsubsection{Conformally trivial case}
If $p$ is pure trace, then $B$ and $C$ are both zero, so either $s=w=0$ or $u=t=0$.  Suppose that $s=w=0$.  Then the remaining equation is the quartic $A^2=4(s^2+t^2)^2$, which has complex solutions for $s^2+t^2$.

Summarizing, recalling that the energy density is $E=2(A+w^2)$,
\begin{theorem}
  In the conformally trivial case ($B = C = 0$), there is an infinity of solutions, $(s, t, u)$, unless both $A$ and $w$ are non-zero, in which case the only solutions are $(s, t, u) = (\pm \sqrt{-E/2}, 0, 0)$.
\end{theorem}

\subsubsection{Conformally non-trivial case}
Next, we consider the number of solutions when $(B,C)\ne(0,0)$.  Denote by $|\cdot|$ the determinant of a $2\times 2$ matrix.  We have
\begin{equation}\label{detp}
  |p| = A^2 - B^2 - C^2 = \left(s^2 + w^2   -  t^2 -u^2\right)^2.
\end{equation}
Also, put
\[ F = (A + 2w^2)^2 - |p| = 4w^4 + 4Aw^2 + B^2 + C^2.\]
\begin{theorem}
 When $(B, C) \ne (0, 0)$, the number of solutions for $(s, t, u)$ is finite,  between one and four.   Each solution is determined uniquely once $s$ is known and $(s, t, u)$ is real if and only if $s$ is real.  The solutions for $s$  are all non-real complex numbers if and only if $|p| < 0$; on the other hand all solutions are real if and only $|p| \ge 0$,  $F\ge 0$ and $A \ge w^2$.
\end{theorem}
\begin{proof}
By assumption, we have $(B, C) \ne (0, 0)$, so $0 < B^2 + C^2  = 4(s^2 + z^2)(t^2 + u^2)$. So we have $z = s^2 + w^2 \ne 0$, so $t$ and $u$ are determined directly in terms of $s$: 
\begin{equation}\label{tuBC}
  t =  \frac{-(Bw + Cs)}{2z}, \hspace{7pt} u =  \frac{-Bs + Cw}{2z}.
\end{equation}
Note that if $s$ is real, then so are $z$, $t$ and $u$.

Eliminating $t,u$ from the Sachs equations, we have the quadratic equation in $s^2$, with real coefficients:
\begin{equation}\label{squartic}
  4s^4+4(A+2w^2)s^2 + (B^2+C^2+4Aw^2+4w^4) = 0
\end{equation}
whose discriminant is $16|p|$.  Thus if $|p|<0$, there are four distinct complex roots.  All roots are real if and only if $|p|\ge 0$ and $A+2w^2\le -\sqrt{|p|}$, i.e., $(A+2w^2)^2 - |p|\ge 0$.
\end{proof}

Equation \eqref{squartic} can be rewritten in the variable $x=s^2+w^2$ as
\begin{equation}\label{squadratic}
  4x^2 + 4Ax + B^2+C^2=4x^2 + 2\op{tr}(p)\,x  + 2^{-1}\op{tr} (\tilde p^2)=0
\end{equation}
where $\tilde p$ is the trace free part of $p$.  This quadratic equation is also satisfied by the square of the trace-free part of $S$, which we show in Lemma \ref{QuadraticRootsLemma}.







\subsection{Symplectic group orbit}\label{GroupOrbit2x2}
Now, supposing that we have a constant solution $S$ to the Sachs equation, we now put $S=s+\Sigma_\circ$ where $\Sigma_\circ$ is symmetric tracefree and $s$ is a scalar, and $\Sigma = \Sigma_\circ + \omega$.  We solve the equation $\dot L(u) = (s+\Sigma) L(u)$, with $L$ invertible.  We may take $L(u)=e^{(s+\Sigma) u}$.  Given $L(u)$, we then solve the equation $\dot H(u) = (L^T(u)L(u))^{-1} = L(-u)L(-u)^T$, with $H(u)$ symmetric.  Then $H(u)$ represents the orbit of the symplectic group, discussed in the last section.

We use the fact that the algebra $\mathcal A$ of $2\times 2$ matrices is the central algebra of split quaternions, with generators
\[J = \begin{bmatrix}0&-1\\1&0\end{bmatrix},\quad K=\begin{bmatrix}1&0\\0&-1\end{bmatrix},\]
\[JK=-KJ,\quad J^2=-1,\quad K^2=1.\]
A split quaternion $X=a + b\,J + c\,K +d\,JK$, where $a,b,c,d$ are (possibly complex) numbers is called {\em trace-free} if $a=0$.  Note that for a trace-free quaternion $X$, we have $X^T = JXJ$.  Also, if $X$ is trace-free, then $X^2$ is central.

We shall often deal in this section with analytic functions of split-quaternionic parameters.  By an {\em entire function} on $\mathbb C$ (or, more generally, $\mathbb C^n$), we understand has its ordinary meaning as a complex-valued function $f(w)$ of $w\in\mathbb C^n$ having a globally convergent Taylor series.  Note that if $f(w)$ vanishes identically on a hyperplane $a.w=0$, then $f(w) = (a.w)g(w)$ for some entire function $g(w)$, which we will then write as $f(w)/(a.w)$.

Decompose the exponential as
\[\exp(z) = \cosh z + z\sigma(z)\]
where $\cosh z$ and $\sigma(z) = z^{-1}\sinh z$ are even entire functions.  We have
\[d\cosh z = z\sigma\,dz, \quad d(z\sigma) = \cosh z\,dz.\]
We write $\cosh z = 1 + 2^{-1}z^2\gamma(z)$, where again $\gamma$ is even and entire.  We have
\[\cosh^2(z) = \frac12(\cosh(2z)+1),\quad \sigma^2(z) = \gamma(2z),\quad \cosh(z)\sigma(z)=\sigma(2z)\]

Note that $\exp(JxJ) = \exp(-J^{-1}xJ)=J^{-1}\exp(-x)J=-J\exp(-x)J$.  Put $L=\exp(u(s+\Sigma))$, so that
\[\dot L = (s+\Sigma) L.\]
We have to integrate
\begin{align*}
  \dot H &= \exp(-u(S+\omega))\exp(-u(S+\omega)^T)\\
         &= \exp(-2us)\exp(-u\Sigma)\exp(-uJ\Sigma J)\\
         &= -\exp(-2us)\exp(-u\Sigma)J\exp(u\Sigma)J\\
         &= -\exp(-2us)\left(\cosh(u\Sigma) - u\Sigma\sigma(u\Sigma)\right)J\left(\cosh(u\Sigma) + u\Sigma\sigma(u\Sigma)\right)J\\
         &=  \exp(-2us)\left(\cosh^2(u\Sigma) + u^2\sigma^2(u\Sigma)\Sigma J\Sigma J + u\cosh(u\Sigma)\sigma(u\Sigma)[\Sigma,J]J\right).
\end{align*}
With the initial condition $H(0)=0$, we have
\[H = e^{-2su}\left(A + B\Sigma J\Sigma J + C[\Sigma,J]J\right)\]
where
\begin{align*}
  A &= \frac{1}{4s(s^2-\Sigma^2)}\left( (\Sigma^2-s^2)(1-e^{2su}) + s^2(e^{2su}-\cosh(2u\Sigma)) - s\Sigma\sinh(2u\Sigma)\right)\\
  B &= \frac{1}{4s\Sigma^2(s^2-\Sigma^2)}\left( -(\Sigma^2-s^2)(1-e^{2su}) + s^2(e^{2su}-\cosh(2u\Sigma)) - s\Sigma\sinh(2u\Sigma)\right)\\
  C &= \frac{1}{4\Sigma(s^2-\Sigma^2)}\left( e^{2su}\Sigma - \Sigma\cosh(2u\Sigma) + s\sinh(2u\Sigma)\right)
\end{align*}
and each quotient is well-defined as an entire function of all variables, and in particular is well-defined if $s=0$, or if $\Sigma$ or $\Sigma\pm s$ is singular.  For example, the numerator of $A$ is an analytic function of $z=s$ and $w=\Sigma$, which vanishes identically at $z=\pm w$, and so is exactly divisible by $z^2-w^2$.  Note that $\Sigma,A,B,$ and $C$ all commute.

Now recall $\Sigma=\Sigma_\circ + \omega$ is the decomposition into symmetric and skew parts, and
\[\Sigma_\circ J + J\Sigma_\circ = \omega J-J\omega = \Sigma_\circ\omega + \omega\Sigma_\circ = 0.\]
We then have
\begin{align*}
  \Sigma J\Sigma J &= (\Sigma_\circ^2-\omega^2) + 2\omega\Sigma_\circ\\
  [\Sigma,J]J &= -2\Sigma_\circ.
\end{align*}
So
\[H = e^{-2su}\left[(A + (\Sigma_\circ^2-\omega^2)B) + 2B\omega\Sigma_\circ - 2C\Sigma_\circ\right].\]

\subsection{Conjugate points}\label{ConjugatePoints2x2}
To find the conjugate points, we compute $|H| = HH^*$:
\begin{align*}
  e^{4su}|H| &= A^2 + 2AB(\Sigma_\circ^2-\omega^2) + (\Sigma_\circ^2-\omega^2)^2B^2 - 4B^2(\omega\Sigma_\circ)^2 - 4C^2\Sigma_\circ^2\\
                &=A^2 + 2AB(\Sigma_\circ^2-\omega^2)+ (\Sigma_\circ^2+\omega^2)^2B^2 - 4C^2\Sigma_\circ^2\\
                &= ( A - B(\Sigma_\circ^2+\omega^2))^2 + 4(AB - C^2)\Sigma_\circ^2\\
                &= ( A - B\Sigma^2)^2 + 4(AB - C^2)\Sigma_\circ^2.
\end{align*}
Now, we have
\begin{align*}
  A - B\Sigma^2 &= \frac{e^{2su}-1}{2s}\\
  4(AB-C^2) &= \frac{s^2-\Sigma^2 + \Sigma^2\cosh(2su) - s^2\cosh(2u\Sigma)}{2s^2\Sigma^2(s^2-\Sigma^2)}e^{2su}
\end{align*}
\begin{align*}
  4s^2\Sigma^2(s^2-\Sigma^2)&\left[(A - B\Sigma^2)^2 + 4(AB - C^2)\Sigma_\circ^2\right] \\
                            &=2e^{2su}\left( (\Sigma^2-s^2)\omega^2 + \Sigma^2(s^2-\omega^2)\cosh(2su) - s^2(\Sigma^2-\omega^2)\cosh(2u\Sigma)\right)\\
                            &=8s^2\Sigma^2e^{2su}\left((s^2-\omega^2)\gamma(2su) - (\Sigma^2-\omega^2)\gamma(2u\Sigma)\right)
\end{align*}
where $\gamma(z) = 2(\cosh z - 1)/z^2$.

From this last display, we conclude at once:
\begin{lemma}\label{ConjugatePointsLemma}
  The real number $u\not=0$ is a conjugate point of $u=0$ for the central null geodesic of the Alekseevsky microcosm $\mathcal G_\alpha(p,\omega)$ if and only if
  \begin{equation}\label{ConjugatePointsLemmaEqn}
    \frac{(s^2-\omega^2)\gamma(2su) - (\Sigma^2-\omega^2)\gamma(2u\Sigma)}{s^2-\Sigma^2} = 0
  \end{equation}
  where $\gamma(z)$ is the even entire function such that $z^2\gamma(z) = 2(\cosh(z)-1)$, and the quotient is exact division of power series.  In particular, when $s^2\ne\Sigma^2$, the condition \eqref{ConjugatePointsLemmaEqn} is
  \[(s^2-\omega^2)\gamma(2su) = (\Sigma^2-\omega^2)\gamma(2u\Sigma).\]
  In the limiting case of $\Sigma^2=s^2\ne 0$, \eqref{ConjugatePointsLemmaEqn} becomes
  \[u(s^2-\omega^2)\cosh(s u)\sinh(s u) + s\omega^2\sinh(s u)^2=0.\]
  Finally, when $\Sigma^2=s=0$, \eqref{ConjugatePointsLemmaEqn} is
  \[2\omega^2u^2=3,\]
  which has no real solutions.
\end{lemma}
  Here as above, the trace-free matrix $\Sigma$ and scalar $s$ are related by $s+\Sigma=S+\omega$ with $S$ a solution to $S^2 - [\omega,S] + p-\omega^2=0$.  Note that we can use the doubling formula $\gamma(2z) = \sigma^2(z)$ to rewrite \eqref{ConjugatePointsLemmaEqn} in terms of $\sigma^2$.

Lemma \ref{ConjugatePointsLemma} expresses the condition for conjugacy in terms of the quantities $s,\Sigma$ (which are determined by the (constant) solution $S$ to the Sachs equation).  It is desirable to reformulate it in terms of $p$ and $\omega$, the original real $2\times 2$ matrices determining the Alekseevsy metric.  To this end, 
\begin{lemma}\label{QuadraticRootsLemma}
  Let
  \[ x = (S-s)^2=\Sigma_\circ^2 = \Sigma^2-\omega^2, \quad y = s^2 -\omega^2.\]
  Then
  \[x + y = -P,\quad 4xy = \tilde p^2.\]
\end{lemma}
\begin{proof}
  We have
  \[ x+y = s^2 + \Sigma^2 - 2\omega^2 = s^2 + \Sigma_\circ^2 - \omega^2 = -P\]
  which follows from the trace part of the Sachs equation $S^2 -[\omega,S] + p - \omega^2=0$, when $s+\Sigma=S+\omega$.  The trace-free part of the Sachs equation gives
  \[2\Sigma_\circ(s+\omega) = 2(s-\omega)\Sigma_\circ = -\tilde p.\]
  So we have
  \begin{align*}
    4xy &= 4\Sigma_\circ^2(s^2-\omega^2) \\
        &= 4\Sigma_\circ(s+\omega)(s-\omega)\Sigma_\circ \\
        &=  \tilde p^2,          
  \end{align*}
  as claimed.
\end{proof}
Lemma \ref{QuadraticRootsLemma} implies that $x$ and $y$ are roots of the quadratic 
\[z^2 + Pz + 4^{-1}\tilde p^2 =0.\]
Summarizing then,
\begin{theorem}
  Let $z=x$ and $z=y$ be roots (with multiplicity) of the quadratic \eqref{squadratic}:
  \begin{equation}\label{ConjugatePointsQuadratic}
    z^2 + Pz + 4^{-1}\tilde p^2=0.
  \end{equation}
  \begin{itemize}
  \item If $|p|\ne 0$, then the points $u$ conjugate to $u=0$ are the (real) solutions to
  \[x\sigma^2(u\sqrt{y+\omega^2}) = y\sigma^2(u\sqrt{x+\omega^2})\]
  where $\sigma(z)=\sinh z/z$.
  \item If $|p|=0$, then $x=y=-P/2$ and the points $u$ conjugate to $u=0$ are the non-zero real solutions to
    \[\frac{uP}{2}\sinh(u\sqrt{\omega^2-P/2})\cosh(u\sqrt{\omega^2-P/2}) = \omega^2\sqrt{\omega^2-P/2}\sinh^2(u\sqrt{\omega^2-P/2}).\]
  \end{itemize}
\end{theorem}

Note that the discriminant of \eqref{ConjugatePointsQuadratic} may be written:
\[(x-y)^2 = |p| = P^2 - \tilde p^2.\]
There are thus no real roots if $|p| < 0$, i.e., $p$ represents an indefinite form.  On the other hand if $|p| >0$, then there is a pair of real roots $x,y = 2^{-1}(-P\pm\sqrt{|p|})$, which have the same sign as each other and sign opposite that of $P$ (since $|p|=P^2-\tilde p^2$).

\begin{theorem}
  In the conformally trivial case of $\tilde p=0$, there is a conjugate point if and only if the energy-momentum $E=P-\omega^2$ is positive.
\end{theorem}
\begin{proof}
  The solutions of \eqref{ConjugatePointsQuadratic} are $y=0$ and $x=-P$.  If $E=0$, then $p=\omega^2$ and the Sachs equation has solution $S=0$, and there are no conjugate points.  If $P=0$ and $\omega\ne 0$, then $S=i\sqrt{|\omega|}$ is a solution of the Sachs equation with conjugate points, and we have $E=-\omega^2>0$.   Finally, if $P\ne 0$ and $E\ne 0$, then condition for conjugacy reads
  \[P\sigma^2(u\sqrt{\omega^2-P}) = 0.\]
  That is,
  \[P\sin(u\sqrt E) = 0\]
  which has a non-zero real solution if and only if $E>0$.
\end{proof}

\begin{theorem}
  If either of the following conditions is true, then there is a conjugate point:
  \begin{itemize}
  \item $|p|<0$; or
  \item the energy-momentum $E=P-\omega^2$ is positive.
  \end{itemize}  
\end{theorem}
\begin{proof}
  If $|p|<0$, then $x$ and $y$ are complex conjugates, and we can choose a branch of the square root such that $\sqrt{x+\omega^2}$ and $\sqrt{y+\omega^2}$ also have this property, and with real part positive.  The condition for conjugacy is thus that $x\sigma^2(u\sqrt{x+\omega^2})$ is {\em real}.  Writing $\sqrt{x+\omega^2} = r + ia$, we have as $u\to\infty$
  \[x\sigma^2(u\sqrt{x+\omega^2}) = \frac{x}{u^2(r+ia)^2}\exp(2u(r+ia))+O(u^{-2})\]
  and the imaginary part of the first term changes sign infinitely often, and so by continuity, the left-hand side is real somewhere.

  If $E>0$, then existence of conjugate points follows by Lemma \ref{ConjugatePointsLemma1}.
\end{proof}

\begin{corollary}
If $p$ has a positive eigenvalue, then there exist conjugate points.
\end{corollary}
\begin{proof}
  If $p$ has a positive and negative eigenvalue, then $|p|<0$.  If $p$ has a positive eigenvalue and a nonnegative eigenvalue, then $E = P-\omega^2 > 0$.
\end{proof}


\bibliographystyle{hplain} 
\bibliography{planewaves} 

\end{document}